\documentclass[aps,twocolumn,showpacs,amsmath,amssymb]{revtex4} 

\usepackage{amsthm}
\usepackage{latexsym}
\usepackage{amsfonts}
\usepackage{bbm,bm,dsfont}
\usepackage{color} 

\newtheorem{theorem}{Theorem}

\newtheorem{example}{Example}

\newcommand{\R}{\mathbb R}

\newcommand{\C}{\mathbb C}

\newcommand{\hi}{\mathcal{H}} 
\newcommand{\li}{\mathcal{L}} 

\newcommand{\tr}[1]{\mathrm{tr}\left[#1\right]} 
\newcommand{\kb}[2]{|#1\,\rangle\langle\,#2|} 

\newcommand{\Mo}{\mathsf{M}} 
\newcommand{\No}{\mathsf{N}} 
\newcommand{\Po}{\mathsf{P}} 
\newcommand{\M}{\mathcal M} 
\newcommand{\A}{\mathsf{A}}

\def\<{\langle}
\def\>{\rangle}
\def\d{{\mathrm d}}
\newcommand{\fii}{\varphi}

\newcommand{\CHI}[1]{\ensuremath{ \chi\raisebox{-1ex}{$\scriptstyle #1$} }}

\newcommand{\ov}{\overline}

\begin{document} 

\title[Complete quantum measurements]{Complete measurements of quantum observables}


%

\author{Juha-Pekka Pellonp\"a\"a}
\email{juhpello@utu.fi}
\address{Turku Centre for Quantum Physics, Department of Physics and Astronomy, University of Turku, FI-20014 Turku, Finland}

\begin{abstract}
We define a complete measurement of a quantum observable (POVM) as a measurement of the maximally refined version of the POVM.
Complete measurements give information from the multiplicities of the measurement outcomes and can be viewed as state preparation procedures.
We show that any POVM can be measured completely by using sequential measurements or maximally refinable instruments. Moreover, the ancillary space of a complete measurement can be chosen to be minimal.
\end{abstract}


\maketitle

%
%
\section{Introduction}

Suppose that we want to perform a measurement of a quantum observable.
The observable can be, e.g.\ the energy (Hamiltonian) $H$ of an atom or the position $Q$ of a particle. The spectrum of $H$ may be degenerate or the particle may have a nonzero spin. Then, intuitively, any measurement of $H$ cannot be seen complete since it does not give `information' about degeneracies of energy states. Similarly, if we measure position $Q$, we do not know the spin of the particle. Obviously, in both cases, a complete measurement  would be a measurement of the maximally refined version of the observable in question since then the outcome space contains also degeneracies or multiplicities of the measurement outcomes. 
Then, the first question is how one can measure maximally refined observables.
As we will see, a solution to this problem is to measure the observable and some other `multiplicity' observable sequentially. For example, one could measure $Q$ first and then a spin observable.

A quantum measurement process starts with the preparation of the system in some  state (density matrix). Then an observable or several observables are measured sequentially. What happens if some of the measurements are complete in the above sense? We will show that a complete measurement can be viewed as a new preparation procedure, i.e.\ the measurement process `ends'  in a complete measurement (and possibly a new process starts if there are measurements left to be performed). 

Any observable allows `preparative' measurements but only for rank-1 observables {\it all} measurements are complete. Hence, if we know that an observable is rank-1 then we know that its measurement can be seen as a preparation of a new measurement even if the specific form of the measurement interaction is not known. Moreover, any observable has a rank-1 refinement so that one can define a complete measurement of an observable as a measurement of its maximally refined version.
Recall that many important observables are already of rank-1:
position and momentum observables (of a spin 0 particle moving on a space manifold), rotated quadratures, phase space observables generated by pure states, the canonical phase observable of a single mode electromagnetic field, and many discrete observables such that Hamiltonians with nondegenerate discrete spectra.

We show that there is a special class of instruments, so-called maximally refinable instruments, related to measurements of an observable which can be directly interpreted as complete measurements. The arise from measurement models of the observable only by changing (refining) the pointer observable of the ancillary system.

We also study minimality of the ancillary spaces of a fixed instrument, or an observable, and show that, to measure an observable, the ancillary space must be at least the space of wave functions related to the observable. The pointer is then the `position' observable.

Before one can study these questions theoretically in the most general context of infinite dimensional Hilbert spaces and continuous outcome spaces, one must define quantum observables and their measurement models rigorously enough. This will be done next.

\section{Observables, instruments, and measurement models}

Let us briefly recall the mathematical description of quantum observables via \emph{normalized positive operator valued measures} (POVMs) \cite{Da,BuLaMi,TeZi,Ho}. Consider a quantum system with a (possibly infinite dimensional) Hilbert space $\hi$ and suppose that the measurement outcomes form a set $\Omega$. 
A POVM is a function $\Mo$ which associates to each (measurable \footnote{Actually, for nondiscrete POVMs, one must also specify a $\sigma$-algebra $\Sigma$ consisting of subsets of $\Omega$. The pair $(\Omega,\Sigma)$ is called an \emph{outcome space} and any $X\in\Sigma$ is called \emph{measurable}.
})
 subset $X\subseteq\Omega$ a positive operator $\Mo(X)$ acting on $\hi$. It is required that for every state (a density matrix) $\varrho$, the mapping
$X\mapsto p^\Mo_\varrho(X)=\tr{\varrho\Mo(X)}$
is a probability distribution. Especially, $\Mo$ satisfies the normalization condition $\Mo(\Omega)=I$ (the identity operator). 
The number $\tr{\varrho\Mo(X)}$ is the probability of getting a measurement outcome $x$ belonging to $X$, when the system is in the state $\varrho$ and a measurement of $\Mo$ is performed. 

A POVM $\Mo$ is called a \emph{projection valued measure} (PVM), a \emph{sharp} POVM, or a {\it spectral measure}, if $\Mo(X)^2\equiv\Mo(X)$. 
In the case $\Omega=\R$, spectral measures correspond to self-adjoint operators.
Sharp POVMs have many important properties. For example, a PVM is always preprocessing clean, i.e.\ it cannot be obtained by (irreversibly) manipulating the state before the measurement and then measuring some other POVM \cite{Pe11}.

Davies and Lewis \cite{DaLe} introduced the concept of instrument which turned out to be crucial in developing quantum measurement theory since,
besides measurement statistics $p^\Mo_\varrho$,
it also describes the conditional state changes $\varrho\mapsto\varrho_X$ due to a quantum measuring process (see also \cite{Da,BuLaMi,TeZi}).
Recall that $\M$ is a {\it Heisenberg instrument} if it associates to each (measurable) set $X$ a completely positive map  \footnote{
Which is ultraweakly continuous (normal) so that $\mathcal I=\M_*$ exists.}
$B\mapsto \M(X,B)$ on the set of bounded operators on $\hi$ and  $X\mapsto \M(X,B)$ is a positive operator valued measure for any $B\ge0$. In addition,
$\M(\Omega,I)=I$ so that $X\mapsto\M(X,I)$ is the {\it associate POVM of $\M$.} We say that $\M$ is {\it $\Mo$-compatible} if $\Mo$ is the associate POVM of $\M$, i.e.\ $\M(X,I)\equiv\Mo(X)$.
Any Heisenberg instrument $\M$ defines a {\it Schr\"odinger instrument} $\mathcal I=\M_*$ by $$
\tr{\M_*(X,\varrho)B}\equiv\tr{\varrho\M(X,B)}
$$ 
and vice versa (i.e.\ $\mathcal I^*=\M$). Especially, 
$p^\Mo_\varrho(X)=\tr{\M_*(X,\varrho)}$, where $\Mo$ is the associate POVM of $\M$,
and one may define a {\it conditional output state} 
$$
\varrho_X=\M_*(X,\varrho)/\tr{\M_*(X,\varrho)}
$$
{\it corresponding to a set} $X$ of outcomes
which describes the state of the subensemble of the measured system in which the outcomes of the measurement lie in $X$.

Recall that
{\it a measurement model}  ${\bf M}$ of a POVM $\Mo$ is a 4-tuple ${\bf M}=\<\hi',\Po,\sigma,U\>$ consisting of a Hilbert space $\hi'$ attached to the probe system, a PVM $\Po$ acting on $\hi'$ (the pointer observable), an initial state $\sigma$ of $\hi'$, and a unitary operator $U$ on $\hi\otimes\hi'$ (the measurement interaction) satisfying the relation
$$
\tr{\varrho\Mo(X)}\equiv\tr{U(\varrho\otimes\sigma)U^*\big(I\otimes\Po(X)\big)}.
$$
Moreover, ${\bf M}$ is {\it pure} if $\sigma=\kb\xi\xi$ for some unit vector $\xi\in\hi'$.
A measurement model ${\bf M}=\<\hi',\Po,\sigma,U\>$ of $\Mo$ defines an $\Mo$-compatible instrument $\M$
by
\begin{equation}\label{memo}
\M_*(X,\varrho)={\rm tr}_{\hi'}\big[U(\varrho\otimes\sigma)U^*\big(I\otimes\Po(X)\big)\big]
\end{equation}
so that $\tr{\varrho\Mo(X)}=\tr{\M_*(X,\varrho)}$.
Ozawa \cite{Oz84} showed that any instrument $\M$ can be {\it realized} as a pure measurement model of the associate POVM $\Mo$ of $\M$, i.e.\ there exists an ${\bf M}=\<\hi',\Po,\sigma=\kb{\xi}{\xi},U\>$ such that \eqref{memo} holds for $\M$.
Recently, the structure of POVMs and instruments (and their measurement models) is thoroughly analyzed by the author \cite{Pe11,Pe12}.

Assume then that we measure POVMs $\Mo$ and $\No$ (of the same Hilbert space $\hi$ but possibly different outcome sets  $\Omega_\Mo$ and $\Omega_\No$) by performing their measurements sequentially (first $\Mo$ and then the next $\No$). This leads to the instrument ${\mathcal J}$ defined by 
$$
{\mathcal J}(X\times Y,B)=\M\big(X,\mathcal N(Y,B)\big),\quad X\times Y\subseteq\Omega_\Mo\times\Omega_\No,
$$ 
where $\M$ and $\mathcal N$ are $\Mo$- and $\No$-compatible instruments describing the measurements.
It defines a sequential joint POVM ${\mathsf J}$ with the value space $\Omega_\Mo\times\Omega_\No$, and whose margins are POVMs 
\begin{eqnarray*}
X&\mapsto&{\mathsf J}(X\times\Omega_\No,I)=\Mo(X), \\
Y&\mapsto&{\mathsf J}(\Omega_\Mo\times Y,I)=\M_1\big(\Omega_\Mo,\No(Y)\big)
\end{eqnarray*}
where the channel $\M\big(\Omega_\Mo,\bullet)$ operates to $\No$, that is, the first measurement disturbs the subsequent one \cite{TeZi}.
Moreover, 
$$
{\mathcal J}_*(X\times Y,\varrho)={\mathcal N}_*\big(Y,\M_*(X,\varrho)\big).
$$
For example, if we measure the position of a particle and then its spin, the measuring process can be viewed as a sequential measurement and the observable measured is the corresponding sequential joint POVM.

In the next section, we study the properties of discrete POVMs and their measurement models. In Section \ref{section3}, we generalize the results obtained in Section \ref{section2} to the general (nondiscrete) case.

\section{Discrete POVMs}\label{section2}

In the case of a discrete POVM $\Mo$ one may choose 
$\Omega=\{x_1,x_2,\ldots\}$ such that $\Mo_i=\Mo(\{x_i\})\ne 0$. Note that
the number $\#\Omega$ of the elements of $\Omega$ and $\dim\hi$ can be infinite.  For all $X\subseteq\Omega$,
$$
\Mo(X)=\sum_{i\,(x_i\in X)}\Mo_i,
$$
and for any $i$ there exists a linearly independent vectors $\{d_{ik}\}_{k=1}^{m_i}$ 
and $\{g_{ik}\}_{k=1}^{m_i}$ of $\hi$ such that
$$\boxed{
\Mo_i=\sum_{k=1}^{m_i}|d_{ik}\>\<d_{ik}|
}
$$
where $m_i$ is the {\it rank} of the effect $\Mo_i$ or the {\it multiplicity} of the outcome $x_i$ and vectors $g_{ik}$ satisfy the biorthogonality condition $\<d_{ik}|g_{i\ell}\>=\delta_{k\ell}$  \cite{Pe11,HyPeYl}.


We say that $\Mo$ is of {\it rank 1} if $m_i\equiv 1$. Any $\Mo$ can be {\it maximally refined} into rank-1 POVM $\Mo^1$ whose value space $\Omega_\Mo$ consists of pairs $(x_i,k)$ where $x_i\in\Omega$ and $1\le k<m_i+1$.
Now 
$$
\boxed{
\Mo^1_{ik}=\Mo^1\big(\{(x_i,k)\}\big)=|d_{ik}\>\<d_{ik}|.
}
$$
Hence, $\Mo$ can be trivially obtained from $\Mo^1$ by {\it relabeling} of outcomes $(x_i,k)$ of $\Mo^1$, i.e.\ by giving the same label $x_i$ to all outcomes $(x_i,k)$, $k=1,2,\ldots$ \cite{HaHePe12}.

By Theorem \ref{theorem1} of Section \ref{section3}, an arbitrary $\Mo$-compatible (Heisenberg) instrument $\M$ is always of the form
$$
\M(X,B)=\sum_{i\,(x_i\in X)}\M_i(B)
$$
with $\M_i(B)=\M(\{x_i\},B)$ and
$$
\boxed{
\M_i(B)
=\sum_{k,l=1}^{m_i}\sum_{s=1}^{r_i}
\<\fii_{iks}|B\fii_{ils}\>|d_{ik}\>\<d_{il}| 
=\sum_{s=1}^{r_i}\A_{is}^*B\A_{is}
}
$$
where vectors $\fii_{iks}\in\hi$ satisfy the condition
$$
\boxed{
\sum_{s=1}^{r_i}\<\fii_{iks}|\fii_{ils}\>=\delta_{kl}
}
$$ 
and $r_i$ is the {\it rank} of the completely positive map $\M_i$ with the (minimal \footnote{That is, for each $i$, operators $\A_{is}$, $s=1,2,\ldots$, are linearly independent, i.e.\  the conditions $\sum_s c_s\fii_{iks}=0$ for all $k$ implies $c_s\equiv0$. (Hence, the rank $r_i$ is minimal.) }) {\it Kraus operators}
$$
\boxed{
\A_{is}=\sum_{k=1}^{m_i}
\kb{\fii_{iks}}{d_{ik}}
}
$$
which completely determine $\Mo$ and $\M$ via vectors $d_{ik}$ and $\fii_{iks}$, respectively. We say that $i$ is the {\it outcome index} of $\Mo$, $k$ is its {\it multiplicity index}, and $s$ is the {\it Kraus rank index} of $\M$. 
Now $\M$ is of {\it rank 1} if $r_i\equiv1$.

The corresponding Schr\"odinger instrument is
$$
\M_*(X,\varrho)=\sum_{i\,(x_i\in X)}\M_{*i}(\varrho)
$$
where $\M_{*i}(\varrho)=\sum_{s=1}^{r_i}\A_{is}\varrho\A_{is}^*$ or
$$
\boxed{
\M_{*i}(\varrho)
=\sum_{k,l=1}^{m_i}\<d_{il}|\varrho|d_{ik}\>\sum_{s=1}^{r_i}|\fii_{ils}\>\<\fii_{iks}|}
$$
so that the probability of getting the outcome $x_i$, when the system is prepared in the state $\varrho$, is
$$
p_\varrho^i=\tr{\varrho\Mo_i}=\tr{\M_{*i}(\varrho)}=\sum_{k=1}^{m_i}\<d_{ik}|\varrho|d_{ik}\>
$$
and the conditional output state corresponding to the set $\{x_i\}$ (or the point $x_i$) is
$$
\varrho_i=\varrho_{\{x_i\}}=\frac1{p_\varrho^i}\sum_{k,l=1}^{m_i}\<d_{il}|\varrho|d_{ik}\>\sum_{s=1}^{r_i}|\fii_{ils}\>\<\fii_{iks}|.
$$

\begin{example}[PVMs]\rm\label{example1}
Suppose that $\Mo$ is a PVM or equivalently $\{d_{ik}\}$ is an orthonormal (ON) basis of $\hi$. Then $\<d_{ik}|d_{jl}\>=\delta_{ij}\delta_{kl}$ and one sees immediately that
$$
\M_i(B)=\Mo_i\Phi(B)\Mo_i
$$
where $\Phi$ is a quantum channel defined by $\Phi(B)=\M(\Omega,B)$. Note that the commutator $[\Phi(B),\Mo_i]=0$. \qed
\end{example}

\begin{example}[Rank-1 POVMs]\rm\label{example2}
Assume that $\Mo$ is rank-1, i.e.\ $m_i\equiv1$. Then, by denoting
$d_i=d_{i1}$ and defining a rank-$r_i$ state $\sigma_i=\sum_{s=1}^{r_i}\kb{\fii_{i1s}}{\fii_{i1s}}$ we get
$$
\M_i(B)
=\tr{B\sigma_i}|d_{i}\>\<d_{i}|=\tr{B\sigma_i}\Mo_i
$$
and
$
\M_{*i}(\varrho)
=\<d_{i}|\varrho|d_{i}\>\sigma_i=p^i_\varrho\sigma_i
$
so that the post measurement state $\varrho_i=\sigma_i$ for all states $\varrho$, i.e.\ 
they do not depend on $\varrho$.
Physically this means that the instrument $\M$ describes a {complete} measurement of $\Mo$ in the sense that the posterior or output states $\varrho_i$ are completely known whatever the input state $\varrho$ is. If $x_i$ is registered then the state $\sigma_i$ is obtained, and if one measures other POVM $\No$ after $\Mo$ one gets its measurement outcome probabilities $p_{\sigma_i}^\No(Y)=\tr{\sigma_i\No(Y)}$ and conditional output states which do not depend on $\varrho$. Thus $\M$ can be viewed as a `state preparator' since it prepares states $\sigma_i$ with probabilities $p^i_\varrho=\tr{\varrho\Mo_i}$. One can measure $\No$ in the fixed state $\sigma_j$ simply by preparing the system in some state $\varrho$ and then first measuring $\Mo$ (with the above $\M$), and then selecting only states with correspond to the value $x_j$, i.e.\ one measures $\Mo$ and $\No$ sequentially to get the joint POVM
$\mathsf J$ which gives
$$
\tr{\varrho \mathsf J(\{x_j\}\times Y)}=\tr{\varrho\M_j(\No(Y))}=
p_{\sigma_j}^\No(Y)p_\varrho^j
$$
from where the probabilities $p_{\sigma_j}^\No(Y)$ can be obtained.
Now the conditional states of the joint measurement with the postselection $x=x_j$ are
$$
\varrho_{\{x_j\}\times Y}=\mathcal N_*(Y,\sigma_j)/\tr{\mathcal N_*(Y,\sigma_j)}
$$
where $\mathcal N$ is the instrument implementing $\No$.

Obviously, when $\sigma_i\equiv\sigma$ (a trivial instrument), posterior states $\varrho_i=\sigma$ for all $i$, i.e.\ they do not depend on the measurement of $\Mo$ in any way. This instrument can be obtained from
$\M_{*i}(\varrho)=p_\varrho^i\sigma_i$ 
trivially by adding a `constant' channel which maps any state (especially $\sigma_i$) to $\sigma$. Hence, we have seen that {\it the measurements of rank-1 POVMs really complete the chain of measurements} (and start  new chains by preparing  states $\sigma_i$ by postselection).

It is easy to show that, if {\it all} $\Mo$-compatible instruments $\M$ are of the above form $\M_i(B)=\tr{B\sigma_i}\Mo_i$ (where $\sigma_i$ are states) then $\Mo$ is necessarily a rank-1 POVM \cite{HeWo}.
To conclude, {\it a POVM admits only complete measurements if and only if it is rank-1}.
\qed
\end{example}

\begin{example}[Rank-1 instruments]\rm\label{example3}
Let $\M$ be rank-1, i.e.\ $r_i\equiv1$. By denoting $\fii_{ik}=\fii_{ik1}$
we have $\<\fii_{ik}|\fii_{il}\>=\delta_{kl}$,
$$
\M_i(B)
=\sum_{k,l=1}^{m_i}\<\fii_{ik}|B\fii_{il}\>|d_{ik}\>\<d_{il}|=\A_i^*B\A_i
$$
where $\A_i=\sum_{k=1}^{m_i}|\fii_{ik}\>\<d_{ik}|$
and
$$
\M_{*i}(\varrho)
=\sum_{k,l=1}^{m_i}\<d_{il}|\varrho|d_{ik}\>|\fii_{il}\>\<\fii_{ik}|=\A_i \varrho\A_i^*.
$$
Hence, if $\varrho=\kb{\psi}{\psi}$ is pure then all posterior states $\varrho_i=\|\A_i\psi\|^{-2}\kb{\A_i\psi}{\A_i\psi}$ are pure. This is the characteristic feature of rank-1 instruments.

If both $\Mo$ and $\M$ are of rank 1 then, by the preceding example,
$$
\M_i(B)
=\<\fii_{i}|B\fii_{i}\>|d_{i}\>\<d_{i}|,\quad
\M_{*i}(\varrho)
=\<d_{i}|\varrho|d_{i}\>|\fii_{i}\>\<\fii_{i}|
$$
where $\fii_i=\fii_{i1}$, $\|\fii_i\|=1$.
Now posterior states $\varrho_i=\kb{\fii_i}{\fii_i}$ are pure {\it for all} states $\varrho$
and they do not depend on $\varrho$.
\qed
\end{example}

\begin{example}[Refinable instruments]\rm\label{example4}
From Example \ref{example2} we see that,
for any POVM $\Mo$, every $\Mo^1$-compatible instrument is
of the form
$$
\ov\M^1_{ik}(B)=\tr{B\sigma_{ik}}|d_{ik}\>\<d_{ik}|
$$
where $\sigma_{ik}$ are states.
It defines  an $\Mo$-compatible instrument $\ov\M$, a {\it compression of $\ov\M^1$}, by
$$
\ov\M_{i}(B)=\sum_{k=1}^{m_i}\tr{B\sigma_{ik}}|d_{ik}\>\<d_{ik}|.
$$
Conversely, if $\M$ is an $\Mo$-compatible instrument defined by vectors $\fii_{iks}$ one can define states $\sigma_{ik}=\sum_{s=1}^{r_i}\kb{\fii_{iks}}{\fii_{iks}}$
and an $\Mo^1$-compatible instrument
\begin{eqnarray*}
\ov\M^1_{ik}(B)=\tr{B\sigma_{ik}}|d_{ik}\>\<d_{ik}|
=\sum_{s=1}^{r_i}
\<\fii_{iks}|B\fii_{iks}\>|d_{ik}\>\<d_{ik}| 
\end{eqnarray*}
whose compression
$$
\ov\M_{i}(B)=\sum_{k=1}^{m_i}\sum_{s=1}^{r_i}
\<\fii_{iks}|B\fii_{iks}\>|d_{ik}\>\<d_{ik}| 
$$
is not necessarily $\M$. If $\ov\M=\M$ we say that $\M$ is {\it refinable}.
Thus, {\it $\Mo^1$-compatible instruments correspond to refinable $\Mo$-compatible instruments.}
\qed
\end{example}

\begin{example}[Maximally refinable instruments]\rm\label{example5}
Similarly to $\Mo$ also its compatible instrument $\M$ can be {\it maximally refined} into a rank-1 instrument $\M^1$. Instead of $\Omega_\Mo$, it must be defined on $\Omega_\M$ which contains all pairs $(x_i,s)$, $x_i\in\Omega$, $1\le s<r_i+1$.
Now 
$$
\M^1_{is}(B)=\M^1\big(\{(x_i,s)\},B\big)=\A_{is}^*B\A_{is}
$$
whose associate POVM 
$$
\M^1_{is}(I)=\A_{is}^*\A_{is}=\sum_{a,b=1}^{m_i}
\<\fii_{ias}|\fii_{ibs}\>|d_{ia}\>\<d_{ib}| 
$$
is defined on $\Omega_\M$. We say that $\M$ is {\it maximally refinable} $\Mo$-compatible instrument if the associate POVM of $\M^1$ can be identified with $\Mo^1$, that is, $\Omega_\M=\Omega_\Mo$ and, for any
$(x_i,k)\in\Omega_\Mo$ there exists a unique $(x_i,s_k)\in\Omega_\M$ and
$\M^1_{is_k}(I)=\Mo^1_{ik}$. But this means that $r_i\equiv m_i$ 
and
$$
\<\fii_{ias_k}|\fii_{ibs_k}\>
=\delta_{ak}\delta_{bk},\qquad
\sum_{s=1}^{m_i}\<\fii_{iks}|\fii_{ils}\>=\delta_{kl}
$$
which can hold only if $\fii_{iks}=\delta_{ss_k}\fii_{ik}$ where $\fii_{ik}\in\hi$, $\|\fii_{ik}\|=1$. Clearly, without restricting generality, we may assume that $s_k=k$. Hence, maximally refinable instruments are of the form
\begin{eqnarray*}
\M_i(B)&=&\sum_{k=1}^{m_i}\<\fii_{ik}|B\fii_{ik}\>|d_{ik}\>\<d_{ik}| \\
&=&
\sum_{k=1}^{m_i}\tr{B\,\kb{\fii_{ik}}{\fii_{ik}}}|d_{ik}\>\<d_{ik}|
\end{eqnarray*}
and thus refinable.
They define rank-1 $\Mo^1$-compatible instruments
$$
\M^1_{ik}(B)=\tr{B\,\kb{\fii_{ik}}{\fii_{ik}}}|d_{ik}\>\<d_{ik}|.
$$
Hence, {\it rank-1 $\Mo^1$-compatible instruments correspond to maximally refinable $\Mo$-compatible instruments.} 

Note that the posterior states $\varrho_{ik}=\kb{\fii_{ik}}{\fii_{ik}}$ of the above instrument $\M^1$ are pure for all states $\varrho$ (see Example \ref{example3}). 
\qed
\end{example}

\subsection{Measurement models}
In this subsection, we consider measurement model realizations of $\Mo$-compatible instruments $\M$ with the vectors $d_{ik}$ and $\fii_{iks}$ (see boxed equations above).

Let $\li$ be an infinite dimensional Hilbert space with a fixed ON basis $\{b_s\}_{s=1}^\infty$ and $\hi_n$ an $n$-dimensional Hilbert space spanned by vectors $b_s$, $1\le s\le n$ (and $\hi_\infty=\li$ if $n=\infty$).
Define an ancillary Hilbert space $\hi'$  which consists of sequences
$$
\zeta=(\zeta_j)_{j=1}^{\#\Omega},\qquad\zeta_j\in\hi_{r_j}
$$
such that $\sum_{j=1}^{\#\Omega}\|\zeta_j\|^2<\infty$ and the inner product is $\<\zeta|\zeta'\>=\sum_{j=1}^{\#\Omega}\<\zeta_j|\zeta_j'\>$.
Obviously, the vectors $b_{is}=(b_{is|1},b_{is|2},\ldots)\in\hi'$, $1\le i<\#\Omega+1$, $1\le s<r_i+1$, defined by
$$
b_{is|j}=\delta_{ij}b_s
$$
form an ON basis of $\hi'$, i.e.\
$$
\zeta=\sum_{i,s}\<b_s|\zeta_i\>b_{is},
$$
and $\dim\hi'=\sum_{i=1}^{\#\Omega}r_i$.

Fix any unit vector $\xi\in\hi'$ and define a unitary measurement coupling
 $U:\,\hi\otimes\hi'\to\hi\otimes\hi'$ by 
\begin{eqnarray*}
U(\psi\otimes\xi)&=&
\sum_{i=1}^{\#\Omega}\sum_{s=1}^{r_i}\sum_{k=1}^{m_i}\<d_{ik}|\psi\>\fii_{iks}\otimes b_{is}
\\
&=&
\sum_{i=1}^{\#\Omega}\sum_{s=1}^{r_i}\A_{is}\psi\otimes b_{is}
\end{eqnarray*}
that is,
$
\<\fii\otimes b_{is}|U(\psi\otimes\xi)\>
=\tr{\kb{\psi}{\fii}\A_{is}},
$
and by extending $U$ to the whole space $\hi\otimes\hi'$. Note that the extension is not unique.

Finally, define a pointer PVM
$$
\Po(X)=\sum_{i\,(x_i\in X)} \sum_{s=1}^{r_i}|b_{is}\>\<b_{is}|
$$
of the ancillary space $\hi'$. By denoting the projection $\Po(\{x_i\})=\sum_s|b_{is}\>\<b_{is}|$ by $\Po_i$ one sees that
$$
{\rm tr}_{\hi'}\big[
U(\varrho\otimes\kb{\xi}{\xi})U^*(I\otimes\Po_i)\big]
=\M_{i*}(\varrho)
$$
i.e.\ the pure measurement model ${\bf M}=\<\hi',\Po,\kb\xi\xi,U\>$ of $\Mo$ realizes the $\Mo$-compatible instrument $\M$.
Indeed, the ancillary space $\hi'$ is the `smallest' possible Hilbert space for the (pure) realization of $\M$ \cite{Pe12}. If $\ov{\bf M}=\<\ov\hi',\ov\Po,\kb{\ov\xi}{\ov\xi},\ov U\>$ is another realization of $\M$ there exists an isometry $\hi'\to\ov\hi'$ and hence $\hi'$ can be embedded in $\ov\hi'$. Then $\Po$ can be seen as a projection of $\ov\Po$ to the subspace $\hi'$. The minimal ${\bf M}$ is unique (up to obvious unitary transformations and the choice of the spectrum of $\Po$ which can be, e.g.\ any $\{y_j\}_{j=1}^{\#\Omega}\subset\R$ so that $\sum_j y_j\Po_j$ is a self-adjoint operator).

Suppose then that one wants to measure a POVM $\Mo$, i.e.\ to construct its pure measurement model $\bf M$, such that the ancillary space $\hi'$ is the smallest possible Hilbert space. Then there are no unnecessary degrees of freedom in the measurement. But this means that one needs to find an $\Mo$-compatible instrument $\M$ such that its realization $\bf M$ is minimal. Since  $\dim\hi'=\sum_{i=1}^{\#\Omega}r_i$ {\it the minimal instrument must be of rank 1} (i.e.\ $r_i\equiv1$).

\begin{example}[Rank-1 instruments]\rm\label{example6}
Let $\M$ be rank-1, i.e $\M_i(B)
=\sum_{k,l=1}^{m_i}\<\fii_{ik}|B\fii_{il}\>|d_{ik}\>\<d_{il}|=\A_i^*B\A_i
$ (see Example \ref{example3}).
Since $\hi_{n=1}=\C b_1\cong\C$ we we may choose $\hi'=\ell^2(\Omega)$, the Hilbert space of square summable complex sequences $(c_j)_{j=1}^{\#\Omega}$, and put $b_1=1$ above. The standard ON basis $\{e_i\}_{i=1}^{\#\Omega}$ of $\ell^2(\Omega)$ consists of sequences $e_1=(1,0,0,\ldots)$, $e_2=(0,1,0,\ldots)$, etc.\  Thus, $b_{i1}=e_i$ and the pointer observable is
$
\Po(X)=\sum_{i\,(x_i\in X)} |e_{i}\>\<e_{i}|
$
showing that {\it its `eigenvalues' are not degenerate}.
Moreover, 
$$
U(\psi\otimes\xi)=\sum_{i=1}^{\#\Omega}\sum_{k=1}^{m_i}\<d_{ik}|\psi\>\fii_{ik}\otimes e_i=\sum_{i=1}^{\#\Omega}\A_i\psi\otimes e_i
$$
and, if $\varrho\otimes\kb{\xi}{\xi}$ is the factorized initial state of the compound `object-apparatus' system before the measurement, the entangled state after the measurement interaction is
$$
\omega_\varrho=U(\varrho\otimes\kb{\xi}{\xi})U^*=\sum_{i,j=1}^{\#\Omega}\A_i\varrho\A_j^*\otimes\kb{e_i}{e_j}
$$
whose subsystem states are the following partial traces:
\begin{eqnarray*}
{\rm tr}_{\hi'}[\omega_\varrho]&=&\sum_{i=1}^{\#\Omega}\A_i\varrho\A_i^*=\M(\Omega,\varrho)
=\sum_{i=1}^{\#\Omega}p_\varrho^i\varrho_i,\\
{\rm tr}_{\hi}[\omega_\varrho]&=&\sum_{i,j=1}^{\#\Omega}\tr{\A_i\varrho\A_j^*}\kb{e_i}{e_j}.
\end{eqnarray*}
Clearly, the probability reproducibility condition 
$$
p_\varrho^i=\tr{\varrho\Mo_i}\equiv\tr{\omega_\varrho\big(I\otimes\Po_i\big)}
=\tr{{\rm tr}_{\hi}[\omega_\varrho]\Po_i}
$$
holds. Here $\Mo_i=\A_i^*\A_i$ and $\Po_i=\kb{e_i}{e_i}$ as before.
Often it is assumed (the projection postulate) that
the object-apparatus state after registering and reading a value $x_i$ of $\Po$ is
\begin{eqnarray*}
\omega_\varrho^i =\frac1{p_\varrho^i}(I\otimes\Po_i)\omega_\varrho(I\otimes\Po_i) 
= \frac1{p_\varrho^i}\A_i\varrho\A_i^*\otimes\kb{e_i}{e_i}
\end{eqnarray*}
with the subsystem states
\begin{eqnarray*}
{\rm tr}_{\hi'}[\omega_\varrho^i]
&=&\frac1{p_\varrho^i}\A_i\varrho\A_i^*=\frac1{p_\varrho^i}\M_{i*}(\varrho)=\varrho_i, \\
{\rm tr}_{\hi}[\omega_\varrho^i]
&=&\kb{e_i}{e_i}
\end{eqnarray*}
i.e.\ the `state of the apparatus has collapsed into the eigenstate' $\kb{e_i}{e_i}$,
the object system is in the posterior state $\varrho_i$ determined by the instrument $\M$ and the total state is factorized.


%
%

If $\Mo$ is a PVM then $U$ is determined by 
$$
U(d_{ik}\otimes\xi)=\fii_{ik}\otimes e_i.
$$
The choice $\fii_{ik}=d_{ik}$ gives
the {\it von Neumann-L\"uders (vN-L) instrument} 
$$
\M(B)=\Mo_iB\Mo_i,\qquad
\M_{*i}(\varrho)
=\Mo_i\varrho\Mo_i
$$
since $\A_{i}=\Mo_i$. (Now $\Phi(B)=B$ in Example \ref{example1}.)
The vN-L instruments have many important properties. For instance, they are ideal and strongly repeatable \cite{BuLaMi}.
Recall that an instrument $\M$ is strongly repeatable if $\M_i\big(\M_j(B)\big)\equiv\delta_{ij}\M_i(B)$, that is, the repetition of $\M$ does not lead to a new result.

If both $\Mo$ and $\M$ are rank-1 then it follows that the interaction
\begin{equation}\label{sdfasdfasdf}
U(\psi\otimes\xi)=\sum_{i=1}^{\#\Omega}\<d_{i}|\psi\>\fii_{i}\otimes e_i
\end{equation}
where $\fii_i=\fii_{i1}$, $\|\fii_i\|=1$ (see Example \ref{example3}). 
\qed
\end{example}

\begin{example}[Maximally refinable instruments]\rm\label{example7}

Let $\M$ be a maximally refinable $\Mo$-compatible instrument (see Example \ref{example5}), i.e.\
$r_i\equiv m_i$, $\fii_{iks}=\fii_{ik}\delta_{ks}$, $\|\fii_{ik}\|=1$, and
$$
\M_i(B)=\sum_{k=1}^{m_i}\<\fii_{ik}|B\fii_{ik}\>|d_{ik}\>\<d_{ik}|.
$$
Note that generally now $\<\fii_{ik}|\fii_{il}\>\ne\delta_{kl}$ and, if
$\M$ is also rank-1 then $\Mo$ must be rank-1. Comparing 
$$
U(\psi\otimes\xi)
=\sum_{i=1}^{\#\Omega}\sum_{k=1}^{m_i}\<d_{ik}|\psi\>\fii_{ik}\otimes b_{ik}
$$
to Eq.\ \eqref{sdfasdfasdf} we see that $U$ can be used for measuring both $\Mo$-compatible $\M$ and an $\Mo^1$-compatible rank-1 instrument 
$$
\M^1_{ik}(B)=\tr{B\,\kb{\fii_{ik}}{\fii_{ik}}}|d_{ik}\>\<d_{ik}|.
$$
The only difference is in the pointer observables.
For $\M$ the pointer observable is $\Po_i= \sum_{k=1}^{r_i}|b_{ik}\>\<b_{ik}|$
whereas for $\M^1$ it is $\Po^1_{ik}=|b_{ik}\>\<b_{ik}|$, the maximally refined $\Po$.
In conclusion, {\it if one can realize some maximally refinable $\Mo$-compatible instrument (and hence measure $\Mo$) as a pure measurement model $\bf M$ then one can measure the maximally refined $\Mo^1$ only by changing the pointer PVM of $\bf M$, i.e.\ by `reading' the multiplicities $k$ of the measurement outcomes $x_i$.}

Choose then $\fii_{ik}=g^1_{ik}=g_{i k}/\|g_{i k}\|$ where vectors $g_{ik}$ satisfy the biorthogonality condition $\<d_{ik}|g_{i\ell}\>=\delta_{k\ell}$. (In the case of a PVM, $g_{ik}=d_{ik}=g^{1}_{ik}$.) Now
$$
\M_{*i}\big(\kb{g^1_{i\ell}}{g^1_{i\ell}}\big)=\|g_{i\ell}\|^{-2}\kb{g^1_{i\ell}}{g^1_{i\ell}}
$$
so that, if the system is prepared in the pure state $\varrho_{i\ell}=\kb{g^1_{i\ell}}{g^1_{i\ell}}$, one gets an outcome $x_i$ with the probability 
$$
p_{\varrho_{i\ell}}^i=\sum_{k=1}^{m_i}\<d_{ik}|\varrho_{i\ell}|d_{ik}\>=\|g_{i\ell}\|^{-2},
$$
and if $x_i$ is registered then the output state is the same $\varrho_{i\ell}$.
Hence we have obtained a kind of `very weak repeatability condition.'
Note that now
$
\A_{is}=
\kb{g^1_{is}}{d_{is}}
$
and $\A_{is}\A_{it}=\delta_{st}\|g_{is}\|^{-1}\A_{is}$.
\qed
\end{example}

\subsection{Complete sequential measurements} 

As we have seen in Example \ref{example7}, it is possible to measure $\Mo^1$ of a POVM $\Mo$ by using a measurement model of $\Mo$ and by changing the pointer observable $\Po$. Next we show that $\Mo^1$ can be measured by first measuring $\Mo$ and then performing the vN-L measurement of a discrete PVM.

Let $\Mo_i$ be a POVM with vectors $d_{ik}$, $k<m_i+1$, and $K$ the largest multiplicity $m_i$ (or $\infty$ if $\sup_i\{m_i\}=\infty$).
Let $\{\No_k\}_{k=1}^K$ be any projections such that $\No_k\No_\ell=\delta_{k\ell}\No_k$ so that they form a PVM $\No$ if one defines $\No_0=I-\sum_{k=1}^K\No_k$. Note that projections $\No_k$ can be any projections associated to some PVM (or a self-adjoint operator) 
whose spectrum contains a discrete part which is large enough (the spectrum may also contain a `continuous' part). 
Define the vN-L instrument $B\mapsto\mathcal N_k(B)=\No_k B \No_k$ implementing the PVM $\No$.

Suppose then that one measures $\Mo$ and $\No$ sequentially.
Since any $\Mo$-compatible instrument is of the form
\begin{eqnarray*}
\M_i(B)
=\sum_{k,l=1}^{m_i}\sum_{s=1}^{r_i}
\<\fii_{iks}|B\fii_{ils}\>|d_{ik}\>\<d_{il}| 
\end{eqnarray*}
the joint instrument is
\begin{eqnarray*}
{\mathcal J}_{ik}(B)&=&\M_i\big(\mathcal N_k(B)\big) \\
&=&
\sum_{a,b=1}^{m_i}\sum_{s=1}^{r_i}\<\No_k\fii_{ias}|B\No_k\fii_{ibs}\>|d_{ia}\>\<d_{ib}| .
\end{eqnarray*}
On the other hand, any $\Mo^1$-compatible instrument is of the form
$$
\ov\M^1_{ik}(B)=\tr{B\sigma_{ik}}|d_{ik}\>\<d_{ik}|
$$
where $\sigma_{ik}$ are states (see Example \ref{example4}).
Now
$$
{\mathcal J}_{ik}(B)\equiv\ov\M^1_{ik}(B)
$$
exactly when
$$
\boxed{
\sum_{s=1}^{r_i}|\No_k\fii_{ibs}\>\<\No_k\fii_{ias}|\equiv\delta_{ak}\delta_{bk}\sigma_{ik}
}
$$
so that we must have $\sigma_{ik}=\No_k\sigma_{ik}\No_k$, i.e.\ $\sigma_{ik}$ can be viewed as a state of the subspace $\No_k\hi$.

Choose Kraus ranks $r_i$ and vectors such that $\fii_{iks}\in\No_k\hi$ for all $i,\,s$,
$\sum_{s=1}^{r_i}\|\fii_{iks}\|^2=1$, and 
$\sum_s c_s \fii_{iks}=0$ for all $k$ implies that $c_s\equiv0$.
By choosing states $\sigma_{ik}=\sum_{s=1}^{r_i}\kb{\fii_{iks}}{\fii_{iks}}$ we get
${\mathcal J}_{ik}(B)\equiv\ov\M^1_{ik}(B)$. Obviously such vectors always exist:

For example, pick a unit vector $\phi_k$ from $\No_k\hi$ (implying $\<\phi_k|\phi_l\>=\delta_{kl}$) and define $\fii_{iks}=c_{iks}\phi_k$ where $\sum_{s=1}^{r_i}|c_{iks}|^2=1$ 
and $\sum_s c_s c_{iks}=0$ for all $k$ implies that $c_s\equiv0$.
Then $\sigma_{ik}=\kb{\phi_k}{\phi_k}$ and $\ov\M^1$ is rank-1.
We have two interesting special cases: 
\begin{enumerate}

\item $\M$ is rank-1 so that $\fii_{ik1}=\fii_{ik}=\phi_k$, i.e.\ $c_{ik1}=1$.
Then 
$
\M_{*i}(\varrho)
=\sum_{k,l=1}^{m_i}\<d_{il}|\varrho|d_{ik}\>|\phi_{l}\>\<\phi_{k}|
$
and 
$$
\mathcal N_{*k}\big(\M_{*i}(\varrho)\big)=\mathsf N_k\M_{*i}(\varrho)\mathsf N_k=\<d_{ik}|\varrho|d_{ik}\>|\phi_{k}\>\<\phi_{k}|
$$ 
describes a complete measurement of $\Mo$ where, after the measurements of $\Mo$ and $\No$, the state of the system is collapsed to the eigenstate $|\phi_{k}\>\<\phi_{k}|$ if $k$ is registered. The probability of getting $(x_i,k)$, $k<m_i+1$, is $\<d_{ik}|\varrho|d_{ik}\>$ (and 0 if $k>m_i$).

\item $\M$ is maximally refinable, that is, $r_i=m_i$ and $\fii_{iks}=\delta_{ks}\fii_{ik}=\delta_{ks}\phi_{k}$ or $c_{iks}=\delta_{ks}$. Now
$
\M_i(\varrho)=\sum_{k=1}^{m_i}\<d_{ik}|\varrho|d_{ik}\>|\phi_{k}\>\<\phi_{k}| 
$
giving the same complete measurement
$
\mathcal N_{*k}\big(\M_{*i}(\varrho)\big)=\<d_{ik}|\varrho|d_{ik}\>|\phi_{k}\>\<\phi_{k}|
$
as above.
\end{enumerate}

\section{Arbitrary POVMS}\label{section3}

In this section, we assume that $\Mo$ is an arbitrary POVM with an arbitrary value space $\Omega$ and generalize the results of the preceding section for $\Mo$.
It is shown in \cite{Pe11,HyPeYl} that 
$$
\Mo(X)=\int_X \sum_{k=1}^{m(x)}|d_k(x)\>\<d_k(x)|\d\mu(x)
$$
where, for all $x$, generalized vectors $\{d_k(x)\}_{k=1}$ are linearly independent \footnote{Note that if $\dim\hi=\infty$ then $d_k(x)$ may belong to a larger space than $\hi$ \cite{Pe11}.}
and $\mu$ is some positive measure. It can always be chosen to be a probability measure $\mu(X)=\tr{\varrho_0\Mo(X)}$ where $\varrho_0$ is any state with (only) positive eigenvalues. We say that vectors $d_k(x)$ are {\it generalized coherent states} of $\Mo$ \cite{HePe12}. Note that $m(x)\le\dim\hi$.
In addition, there are (linearly independent) vectors $\{g_\ell(x)\}_{\ell=1}^{m(x)}$ such that $\<d_k(x)|g_\ell(x)\>=\delta_{k\ell}$ \cite{HyPeYl}.
If $\Mo$ is a PVM then {\it formally} $\<d_k(x)|d_\ell(y)\>=\delta_{k\ell}\delta_y(x)$ where $\delta_y(x)$ is the `Dirac's delta' concentrated on $y$.
If $\Mo$ is the spectral measure of a self-adjoint operator $S$ then
$Sd_k(x)=xd_k(x)$ if $x\in\R$ belongs to the spectrum of $S$ \cite{HyPeYl}. 
Hence, we call $m(x)$ the multiplicity of the measurement outcome $x$. 
Note that solutions $d_k(x)$ of $Sd_k(x)=xd_k(x)$ turn out to be useful for determining $\Mo$ of $S$ in many practical situations.

\begin{theorem}\label{theorem1}
Any $\Mo$-compatible instrument $\M$ is of the form
\begin{eqnarray*}
&&\M(X,B)=\int_X\sum_{s=1}^{r(x)}\A_{s}(x)^*B\A_{s}(x)\d\mu(x)\\
&&=\int_X\sum_{k,l=1}^{m(x)}\sum_{s=1}^{r(x)}
\<\fii_{ks}(x)|B\fii_{ls}(x)\>|d_{k}(x)\>\<d_l(x)| \d\mu(x)
\end{eqnarray*}
where the Kraus rank $r(x)$ is minimal, that is, the Kraus operators $\A_s(x)=\sum_{k=1}^{m(x)}\kb{\fii_{ks}(x)}{d_k(x)}$ are linearly independent for any fixed $x\in\Omega$. In addition, $\sum_{s=1}^{r(x)}\<\fii_{ks}(x)|\fii_{ls}(x)\>=\delta_{kl}$.
\end{theorem}

\begin{proof}
Define (possibly unbounded) operators
$\A(x)=\sum_{k=1}^{m(x)}\kb{b_k}{d_k(x)}$
where $\{b_k\}$ can be chosen to be an ON basis of $\hi$.
 From Theorem 3 of \cite{Pe12}
follows that
$$
\M(X,B)\equiv 
\int_X \A(x)^* T_x(B)\A(x)\d\mu(x).
$$
where $T_x$ is a completely positive channel thus having a minimal Kraus decomposition
$$
T_x(B)=
\sum_{s=1}^{r(x)} \A^T_s(x)^*B\A^T_s(x).
$$
By defining $\A_s(x)=\A^T_s(x)\A(x)$ and $\fii_{ks}(x)=\A_s^T(x)b_k$, Theorem follows.
\end{proof}

If $\Mo$ is discrete, i.e.\ concentrated on points $x_i$, then $\mu$ can be chosen to be a sum of Dirac deltas (point measures),
$
\mu=\sum_i \delta_{x_i}.
$
Then, for instance,
\begin{eqnarray*}
\Mo(X)&=&\int_X \sum_{k=1}^{m(x)}|d_k(x)\>\<d_k(x)|\d\mu(x) \\
&=&\sum_{i\,(x_i\in X)} \sum_{k=1}^{m(x_i)}|d_k(x_i)\>\<d_k(x_i)|.
\end{eqnarray*}
By denoting $m_i=m(x_i)$, $d_{ik}=d_k(x_i)$, $r_i=r(x_i)$, $\A_{is}=\A_s(x_i)$, and $\fii_{iks}=\fii_{ks}(x_i)$, we obtain the boxed equations of Section \ref{section2}.
But this works also conversely. Namely, by doing the above replacements and
replacing sums $\sum_i$ by integrals $\int_X(\cdots)\d\mu(x)$ one can generalize all definitions and results of Section \ref{section2} to arbitrary POVMs and instruments (see general results and methods from \cite{Pe12}).
Note that the subindex $i$ is removed and replaced by adding $(x)$.

For example, if $\Mo$ is rank-1, i.e.\ $m(x)=1$ (or 0), then 
$$
\Mo(X)=\int_X |d(x)\>\<d(x)|\d\mu(x)
$$
and every $\Mo$-compatible instrument is of the form
$$
\M(X,B)=\int_X\tr{B\sigma(x)} |d(x)\>\<d(x)|\d\mu(x)
$$
or $\M_*(X,\varrho)=\int_X\sigma(x)\tr{\varrho\Mo(\d x)}$
where $\sigma(x)$ are states.
For any POVM $\Mo$, its maximally refined rank-1 POVM is now
$$
\Mo^1(X\times\{k\})=\int_X |d_k(x)\>\<d_k(x)|\d\mu(x)
$$
where $x\in\Omega$ and $1\le k<m(x)+1$, so that every instrument implementing $\Mo^1$ is of the form 
$$
\ov\M^1(X\times\{k\},B)=\int_X\tr{B\sigma_k(x)}\kb{d_k(x)}{d_k(x)}\d\mu(x),
$$
where $\sigma_k(x)$ are states,
and is rank-1 if and only if $\sigma_k(x)=\kb{\fii_k(x)}{\fii_k(x)}$.
The compression of $\ov\M^1$ is
$$
\ov\M(X,B)=\int_X\sum_{k=1}^{m(x)}\tr{B\sigma_k(x)}\kb{d_k(x)}{d_k(x)}\d\mu(x)
$$
and $\M$ is maximally refinable if 
$$
\M(X,B)
=\int_X\sum_{k=1}^{m(x)}\<\fii_k(x)|B\fii_k(x)\>\kb{d_k(x)}{d_k(x)}\d\mu(x)
$$
corresponding to a rank-1 $\Mo^1$-compatible instrument.
Now $r(x)=m(x)$, $\fii_{ks}(x)=\fii_k(x)\delta_{ks}$, and $\|\fii_k(x)\|=1$ in Theorem \ref{theorem1}.
For example, the choice $\fii_k(x)=g_k(x)/\|g_k(x)\|$ gives
 $\A_s(x)=\|g_s(x)\|^{-1}\kb{g_s(x)}{d_s(x)}$ and
$\A_s(x)\A_t(x)=\delta_{st}\|g_s(x)\|^{-1}\A_s(x)$ as in Example \ref{example7}. 

In the case of a rank-1 $\Mo$ compatible instrument $\M$, 
$r(x)=1$, $\fii_{k1}(x)=\fii_k(x)$, and $\<\fii_k(x)|\fii_l(x)\>=\delta_{kl}$ so that then
\begin{eqnarray*}
\M(X,B)&=&\int_X\A_1(x)^*B\A_1(x)\d\mu(x) \\
&=&\int_X\sum_{k,l=1}^{m(x)}
\<\fii_{k}(x)|B\fii_{l}(x)\>|d_{k}(x)\>\<d_l(x)| \d\mu(x).
\end{eqnarray*}
Next we consider posterior states.

Let $\varrho$ be an initial state of the system before the measurement of $\Mo$ described by 
an $\Mo$-compatible instrument $\M$.
Then the measurement outcome probabilities are
$$
p^\Mo_\varrho(X)=\tr{\varrho\Mo(X)}=\tr{\varrho\M(X,I)}=\tr{\M_*(X,\varrho)}.
$$
Denote by $w_\varrho$ the density (or weight function) of $p^\Mo_\varrho$ with respect to $\mu$, i.e.\ $\d p^\Mo_\varrho(x)=w_\varrho(x)\d\mu(x)$.
If $w_\varrho(x)\ne 0$, then a posterior state  \cite{Oz85,Ho98,Pe12} corresponding to the outcome $x\in\Omega$ is
$$
\varrho_x=
w_\varrho(x)^{-1}\sum_{k=1}^{r(x)}\A_s(x)\varrho\A_s(x)^*.
$$
Sometimes $\varrho_x$ is interpreted as a final state of the system after
the value $x$ is observed.
This interpretation is problematic since, on the first hand, $\varrho_x$ is not necessarily unique. On the other hand, 
it may happen that $\mu(\{x\})=0$ (e.g.\ position observables). Then it is better to
define the {conditional output state} 
$$
\varrho_X=\M_*(X,\varrho)/\tr{\M_*(X,\varrho)}
$$
 corresponding to a set $X$ of outcomes.
Since
$$
\varrho_X=p^\Mo_\varrho(X)^{-1}\int_X\varrho_x\d p^\Mo_\varrho(x)
$$
is a `continuous' mixture of the states $\varrho_x$ we have an obvious interpretation:

Prepare the system in the fixed state $\varrho$.
Do the measurement to get some value $x^1$ and a (possibly unknown) output or posterior state $\varrho_{x^1}$ corresponding to the outcome $x^1$.
Repeat the process $N$ times to get the sequences $\{x^1,x^2,\ldots,x^N\}$ and
$\{\varrho_{x^1},\varrho_{x^2},\ldots,\varrho_{x^N}\}$.
Let $N_X$ be the number of results $x^j$ which belong to $X\subseteq\Omega$.
Then
$$
\lim_{N\to\infty}\frac{N_X}{N}=p_\varrho^\Mo(X)
$$
and the limit mean value
\begin{eqnarray*}
\lim_{N\to\infty}\frac{1}{N_X}\sum_{x^j\in X}\varrho_{x^j}
&=&\lim_{N\to\infty}
\left(\frac{N_X}N\right)^{-1}\sum^*_{x^j\in X}\varrho_{x^j}\frac{N_{\{x^j\}}}{N} \\
&=&p^\Mo_\varrho(X)^{-1}\int_X\varrho_x\d p^\Mo_\varrho(x)
=\varrho_X
\end{eqnarray*}
where $\sum^*$ means that the sum is taken over distinct values $x^j$ only.
Note that in the discrete case $\varrho_{\{x_i\}}=\varrho_{x_i}=\varrho_i$ if $\Mo(\{x_i\})\ne 0$.

Let ${\bf M}=\<\hi',\Po,\sigma=\kb{\xi}{\xi},U\>$ be a measurement model of a POVM $\Mo$ and $\M$ the related $\Mo$-compatible instrument, i.e.\
$\bf M$ is a pure realization of $\M$.
The structure of measurement models is completely determined in \cite{Pe12} so that we consider only a special case where $\hi'$ is the smallest possible ancillary Hilbert space. It can be shown \cite{Pe12} that it is (unitarily equivalent with) $L^2(\mu)$, the space of (square integrable) wave functions $\Psi:\,\Omega\to\C$, so that we choose $\hi'=L^2(\mu)$. Now $\M$ is rank-1 and
$\Po(X)\Psi=\CHI X\Psi$ (where $\CHI X(x)=1$ if $x\in X$ and 0 otherwise)
so that the pointer PVM $\Po$ is the usual `position observable' on $\Omega$.
Finally, the measurement interaction on $\hi\otimes L^2(\mu)$ is some extension of
$$
\big[U (\psi\otimes\xi)\big](x)=\sum_{k=1}^{m(x)}\<d_k(x)|\psi\>\fii_k(x)\otimes 1
$$
where $\psi\in\hi$. Now vectors $\fii_k(x)$ determine rank-1 $\M$ as before.

To end this section we note that, similarly as in the discrete case, 
the $\Mo$-compatible instruments
$$
\M(X,B)=\int_X\tr{B\sigma(x)}\d\Mo(x)
$$
can be seen as preparators of (statistical mixtures of) states $\sigma(x)$.
Note that $$p_\varrho^\Mo(X)\varrho_X =\M_*(X,\varrho)=\int_X\sigma(x)\d p_\varrho^M(x)=\int_X\varrho_x\d p_\varrho^\Mo(x)$$ so that $\varrho_x=\sigma(x)$ (almost all $x$).
If $\Mo$ is rank-1 then all its instruments (and thus measurement models or measurements) can be seen as such preparators. {\it If we know that $\Mo$ is rank-1 then we know that any measurement of $\Mo$ gives (approximately) some fixed output state $\sigma(x)$ by post selection whatever the input state $\varrho$ is.}
The final task is to show that, for any POVM $\Mo$, we can measure its maximally refined rank-1 version $\Mo^1$. Similarly as in the discrete case one can measure $\Mo^1$ by using a measurement model realizing maximally refinable $\Mo$-compatible instrument $\M$. One just changes the pointer observable $\Po$ to its maximally refined $\Po^1$ or performs a sequential measurement  with a discrete self-adjoint operator as before. Next we consider the last option in the case of a rank-1 instrument $\M$.

Let $K$ be the largest multiplicity $m(x)$ of $\Mo$ and $\{\No_k\}_{k=0}^K$ some PVM. Pick some unit vectors $\phi_k\in\No_k\hi$ and define an
$\Mo$-compatible instrument $\M(X,B)=\int_X\sum_{k,l=1}^{m(x)}\<\phi_k|B\phi_l\>\kb{d_k(x)}{d_l(x)}{\d}\mu(x)$.
Then
\begin{eqnarray*}
\mathcal J(X\times\{k\},B)&=&\M(X,\No_kB\No_k) \\
&=&\int_X\<\phi_k|B\phi_k\>\kb{d_k(x)}{d_k(x)}\d\mu(x)
\end{eqnarray*}
implements $\Mo^1$, that is,
$$
\Mo^1(X\times\{k\})
=\int_X\kb{d_k(x)}{d_k(x)}\d\mu(x)=\mathcal J(X\times\{k\},I).
$$
Note that, if the `multiplicity' $k$ is obtained in the last  vN-L measurement of $\No$ then the input state $\varrho$ is collapsed into $\phi_k$ and the measurement is completed.

\section{Conclusion}

We have seen that the maximally refined version $\Mo^1$ of a POVM $\Mo$ can be measured by using a (maximally refinable) measurement model for $\Mo$ and then changing the pointer observable. Another way to measure $\Mo^1$ is to 
choose a quite arbitrary discrete PVM $\No$ and a suitable measurement model of $\Mo$ and then perform a sequential measurement of $\Mo$ and $\No$.
The instrument $\M$ of the first measurement can be chosen to be rank-1, i.e.\ it is minimal.

We call a measurement of $\Mo^1$ complete since
it also gives information from multiplicities of the outcomes of a POVM $\Mo$.
Moreover, a complete measurement can be viewed as a state preparation procedure.

Although the maximally refinable instrument of $\Mo$ can be interpreted as a minimal (i.e.\ rank-1) instrument of $\Mo^1$ it is not a minimal instrument of $\Mo$ (if $r_i=n_i>1$). Hence, considered as instruments of $\Mo$, the realizations of maximally refinable instruments have unnecessarily large ancillary Hilbert space $\hi'$. This problem can be overcome by measuring $\Mo$ and $\No$ sequentially.
If $\Mo$ is of rank 1 then $\Mo^1=\Mo$ and its maximally refinable instruments are automatically minimal.

\end{document}